\documentclass{article}

\usepackage{amssymb}
\usepackage{amsthm}
\usepackage{graphicx}
\usepackage{subcaption}
\usepackage{wrapfig}
\usepackage{xspace}
\usepackage{times}

\usepackage[a4paper]{geometry}

\newtheorem{theorem}{Theorem}
\newtheorem{lemma}{Lemma}
\newtheorem{corollary}{Corollary}

\newcommand{\kpartition}[1]{degree-${#1}$ edge partition\xspace}

\graphicspath{{./figures/}}

\begin{document}

\title{New Results on Edge Partitions of 1-plane Graphs
\thanks{NSERC funding is gratefully acknowledged for WE and SW.}}

\author{Emilio Di Giacomo\thanks{Universit{\`a} degli Studi di Perugia, Italy, \texttt{\{name.surname\}@unipg.it}} \and Walter Didimo\footnotemark[2] \and William S. Evans\thanks{University of British Columbia, Canada, \texttt{will@cs.ubc.ca}} \and Giuseppe Liotta\footnotemark[2] \and Henk Meijer\thanks{University College Roosevelt, The Netherlands, \texttt{h.meijer@ucr.nl}} \and Fabrizio Montecchiani\footnotemark[2] \and Stephen K. Wismath\thanks{University of Lethbridge, Canada, \texttt{wismath@uleth.ca}}}

\date{}

\maketitle

\begin{abstract}
A $1$-plane graph is a graph embedded in the plane such that each edge is crossed at most once. A NIC-plane graph is a $1$-plane graph such that any two pairs of crossing edges share at most one end-vertex.  An edge partition of a $1$-plane graph $G$ is a coloring of the edges of $G$ with two colors, red and blue, such that both the graph induced by the red edges and the graph induced by the blue edges are plane graphs. We prove the following:
$(i)$ Every  NIC-plane graph admits an edge partition such that the red graph has maximum vertex degree three; this bound on the vertex degree is worst-case optimal. 
$(ii)$ Deciding whether a $1$-plane graph admits an edge partition such that the red graph has maximum vertex degree two is NP-complete.
$(iii)$ Deciding whether a $1$-plane graph admits an edge partition such that the red graph has maximum vertex degree one, and computing one in the positive case, can be done in quadratic time.
Applications of these results to graph drawing are also discussed.
\end{abstract}

\section{Introduction}

Partitioning the edges of a graph such that each partition set induces a subgraph with special properties is a fundamental problem in graph theory, with applications in graph algorithms and graph drawing. In 1988, Colbourn and Elmallah~\cite{ec-pepg+-88} showed that the edge set of every planar graph can be partitioned into two partial $3$-trees. This result was then improved by Kedlaya~\cite{Kedlaya1996238} and by Ding {\em et al.}~\cite{Ding2000221} who showed how to partition the edges of a planar graph into two partial $2$-trees. Gon{\c{c}}alves~\cite{DBLP:conf/stoc/Goncalves05} showed that the edge set of every planar graph can be partitioned into two outerplanar graphs, thus solving a  conjecture by Chartrand, Geller, and Hedetniemi~\cite{Chartrand197112}. Concerning more than two partition sets, Schnyder~\cite{DBLP:conf/soda/Schnyder90} proved that the edges of a maximal planar graph can be colored with three colors so that each partition set forms a spanning tree of the graph. This edge partition can be used to obtain a  planar straight-line drawing of the graph on an integer grid of quadratic size~\cite{DBLP:conf/soda/Schnyder90}.

\begin{figure*}[t]
\begin{subfigure}[b]{.32\linewidth}
\centering
\includegraphics[scale=0.4,page=1]{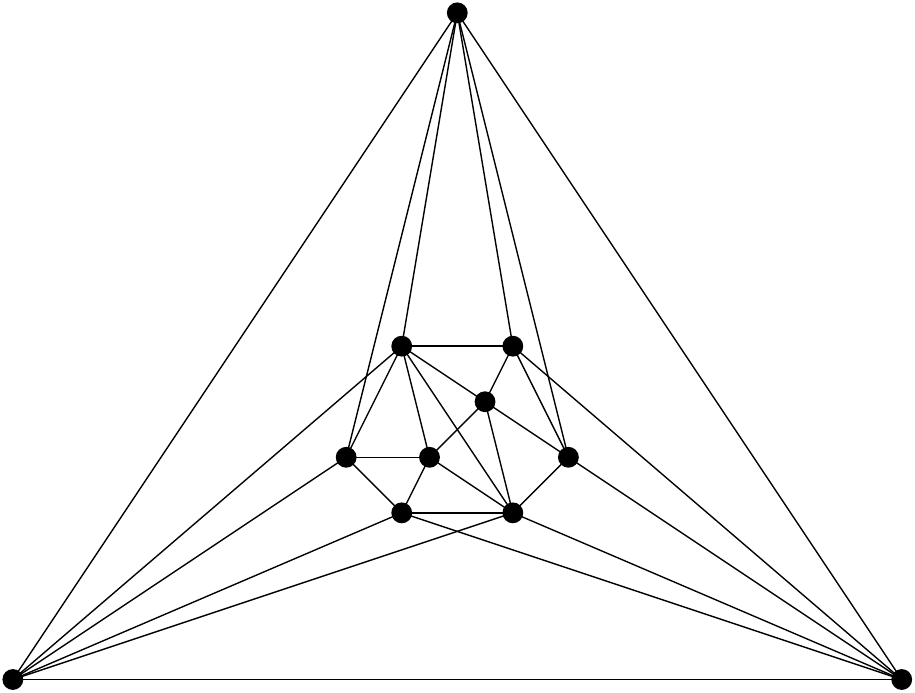}
\caption{}\label{fi:nic}
\end{subfigure}%
\begin{subfigure}[b]{.32\linewidth}
\centering
\includegraphics[scale=0.4,page=2]{figures/example}
\caption{}\label{fi:nic-ep}
\end{subfigure}%
\begin{subfigure}[b]{.32\linewidth}
\centering
\includegraphics[scale=0.23,page=3]{figures/example}
\caption{}\label{fi:opvr}
\end{subfigure}%
\caption{(a) A NIC-plane graph $G$; (b) A \kpartition{3} of $G$ (the  red edges are bold); (c) An ortho-polygon visibility representation $\gamma$ of $G$ computed by exploiting the edge partition shown in (b).}
\end{figure*}

We study edge partitions of \emph{$1$-plane graphs}, i.e., graphs embedded in the plane with at most one crossing per edge. These graphs have been the subject of a rich body of literature, as shown in the annotated bibliography of Kobourov \emph{et al.}~\cite{2017arXiv170302261K}.
An \emph{IC-plane graph} is a $1$-plane graph such that any two pairs of crossing edges do not share an end-vertex (see e.g.~\cite{DBLP:journals/tcs/BrandenburgDEKL16,ks-cpgic-JGT10,DBLP:journals/ipl/LiottaM16,z-dcmgp-AMS14,zl-spgic-CEJM13}), where IC stands for ``independent crossings''. A \emph{NIC-plane graph} is a $1$-plane graph where any two pairs of crossing edges of $G$ share at most one end-vertex (see e.g.~\cite{DBLP:journals/corr/CzapS14,z-dcmgp-AMS14}), where NIC stands for ``near independent crossings''.  An example of NIC-plane graph is shown in Fig.~\ref{fi:nic}. 

An \emph{edge partition}  of a $1$-plane graph $G$ is a coloring of each of its edges with one of two colors, \emph{red} and \emph{blue}, such that both the \emph{red graph} $G_R$ induced by the red edges and the \emph{blue graph} $G_B$ induced by the blue edges are plane graphs. An edge partition such that one of $G_R$ and $G_B$, say $G_R$, has maximum vertex degree $k$ will be called a \emph{\kpartition{k}}. An example of edge partition is shown in Figure~\ref{fi:nic-ep}. 

Czap and Hud\'ak~\cite{DBLP:journals/combinatorics/CzapH13} proved that optimal $1$-plane graphs, i.e., $1$-plane graphs with the maximum number of edges, admit an edge partition  such that the red graph is a forest. Ackerman~\cite{DBLP:journals/dam/Ackerman14} generalized this result to any (non-optimal) 1-plane graph. More recently, Lenhart \emph{et al.}~\cite{DBLP:journals/tcs/LenhartLM17} proved that every optimal $1$-plane graph admits a \kpartition{4} that can be computed in linear time;  the bound on the vertex degree of the red graph is worst-case optimal. The study of {\kpartition{k}}s (with small values of $k$) has interesting applications also in graph drawing. For example, an \emph{ortho-polygon visibility representation (OPVR)} of an embedded graph $G$ maps each vertex to an orthogonal polygon, and each edge $(u,v)$ to a vertical or horizonal uncrossed segment between the  polygons of $u$ and $v$~\cite{ArXiv}. The \emph{vertex complexity} of an OPVR $\gamma$ is the minimum number $k$ such that any polygon of $\gamma$ contains at most $k$ reflex corners. An example of OPVR with vertex complexity four is shown in Figure~\ref{fi:opvr}. Di Giacomo \emph{et al.}~\cite{ArXiv} proved that every $n$-vertex $3$-connected $1$-plane graph that admits a \kpartition{k}, also admits an embedding-preserving OPVR with vertex complexity at most $2k$. Motivated by this result, Di Giacomo \emph{et al.}~\cite{ArXiv} proved that every $3$-connected $1$-plane graph admits a \kpartition{6}, and that this bound on the vertex degree of the red graph is worst-case optimal. 

The contribution of this paper is as follows:

\begin{itemize}

\item We present a linear-time algorithm to compute {\kpartition{3}}s of $3$-connected NIC-plane graphs; this bound on the vertex degree of the red graph is worst-case optimal (Section~\ref{se:deg3}). This result improves the one in~\cite{ArXiv}, for the case of NIC-plane graphs.  As a consequence of our result and of the results in~\cite{ArXiv}, every $3$-connected NIC-plane graph admits an OPVR with vertex complexity at most six, which improves, for the case of NIC-plane graphs, the upper bound on the vertex complexity of $3$-connected $1$-plane graphs proved in~\cite{ArXiv}. Figure~\ref{fi:opvr} shows an OPVR  with vertex complexity four, obtained from the \kpartition{2} in Figure~\ref{fi:nic-ep}.

\item We then turn our attention to the complexity of deciding the existence of {\kpartition{k}}s, when $k <3$. We prove that deciding whether a $1$-plane graph admits a \kpartition{2} is NP-complete, even when the graph is NIC-plane (Section~\ref{se:deg2}). 
On the positive side, we show a quadratic-time algorithm that tests whether a $1$-plane graph admits a \kpartition{1}, and that computes such a partition if it exists (Section~\ref{se:deg1}). 

\end{itemize}

\section{Preliminaries}\label{se:preliminaries}

A \emph{drawing} $\Gamma$ of a graph $G=(V,E)$ is a mapping of the vertices of $V$ to points of the plane, and of the edges of $E$ to Jordan arcs connecting their corresponding endpoints but not passing through any other vertex. We only consider \emph{simple} drawings, i.e., drawings such that two arcs representing two edges have at most one point in common, and this point is either a common endpoint or a common interior point where the two arcs properly cross each other. $\Gamma$ is \emph{planar} if no edge is crossed. A \emph{planar graph} is a graph that admits a planar drawing. A planar drawing of a graph subdivides the plane into topologically connected regions, called \emph{faces}. The infinite region is the \emph{outer face}. A \emph{planar embedding} is an equivalence class of planar drawings that define the same (i.e., topologically equivalent) set of faces and same outer face. A \emph{plane graph} is a planar graph with a given planar embedding.   
The concept of planar embedding is extended to non-planar drawings as follows. Given a non-planar drawing $\Gamma$, replace each crossing with a dummy vertex. The resulting planarized drawing has a planar embedding. An \emph{embedding} of a (non-planar) graph $G$ is an equivalence class of  drawings whose planarized versions have the same planar embedding. An \emph{embedded graph} $G$ is a graph with a given embedding. A \emph{$1$-plane graph} $G$ is an embedded graph with at most one crossing per edge. It is known that $1$-plane graphs with $n$ vertices have at most $n-2$ crossings and $4n-8$ edges~\cite{Suzuki2010}. A $1$-plane graph with $4n-8$ edges is called \emph{optimal}. 
Let $e_1=(u_1,v_1)$ and $e_2=(u_2,v_2)$ be two crossing edges, and let $N(e_1,e_2)=\{u_1,v_1,u_2,v_2\}$. A $1$-plane graph $G$ is a \emph{NIC-plane graph} if any two pairs of crossing edges, $\{e_1,e_2\}$ and $\{e_3,e_4\}$, are such that $|N(e_1,e_2) \cap N(e_3,e_4)| \leq 1$.

\section{Degree-3 edge partitions}\label{se:deg3}
 
In this section we prove that every $3$-connected NIC-plane graph admits a \kpartition{3}. 
Also, we exhibit $3$-connected NIC-plane graphs for which the red graph of any edge partition has maximum vertex degree at least three.  
A $1$-plane graph $G$ is \emph{crossing augmented}, if for each pair of crossing edges $(u,v)$ and $(w,z)$, the subgraph of $G$ induced by $\{u,v,w,z\}$ is $K_4$. The edges of such a $K_4$ distinct from $(u,v)$ and $(w,z)$ form a 4-cycle, and are called the \emph{cycle edges} of $(u,v)$ and $(w,z)$. A $1$-plane graph can always be made crossing augmented in linear time, by adding the missing cycle edges without introducing any new edge crossings (see e.g.~\cite{Suzuki2010}). A pair of crossing edges form a \emph{kite} if the crossing point is inside the 4-cycle formed by the corresponding cycle edges. The proof of the next theorem exploits an argument similar to that used in~\cite{ArXiv} for $3$-connected $1$-plane graphs.

\begin{theorem}\label{th:deg3}
Every $3$-connected NIC-plane graph with $n$ vertices admits a \kpartition{3}, which can be computed in $O(n)$ time. This bound on the vertex degree of the red graph is worst-case optimal.
\end{theorem}
\begin{proof}
We assume that $G$ is crossing augmented, as otherwise it can be augmented in $O(n)$ time~\cite{Suzuki2010}. 

Observe first that, since $G$ is a NIC-plane graph, the cycle edges of $G$ are not crossed. Indeed, let $e_1=(u,w)$ be a cycle edge of two crossing edges $e_3=(u,v)$ and $e_4=(w,z)$. Suppose, for a contradiction, that $e_1=(u,w)$ is crossed by an edge $e_2=(x,y)$ (see also Figure~\ref{fi:deg3-ub}). Then $|N(e_1,e_2) \cap N(e_3,e_4)| \geq 2$, a contradiction since $G$ is NIC. 

From $G$ we remove all crossing edges and suitably orient the remaining edges; we then use the orientation of the cycle edges of a crossing pair to determine which edge of the pair is colored red. Namely, let $G_p$ be the plane graph obtained from $G$ by removing all the crossing edges. As explained above, all the cycle edges of $G$ are in $G_p$.  Let $G_p^+$ be a plane triangulation obtained by edge-augmenting $G_p$. We apply a {\em Schnyder trees decomposition}~\cite{DBLP:conf/soda/Schnyder90} to $G_p^+$. Schnyder~\cite{DBLP:conf/soda/Schnyder90} proved that the internal edges (i.e., those that do not belong to the outer face) of a plane triangulation can be oriented such that each internal vertex has exactly three outgoing edges and the vertices of the outer face have no outgoing edge. We arbitrarily orient the edges of the outer face of $G_p^+$ and we obtain a {\em $3$-orientation} of $G_p^+$, that is an orientation of its edges such that every vertex has at most three outgoing edges. Based on this $3$-orientation, we now show how to color red one edge for each pair of crossing edges so that each vertex is incident to at most three red edges.
Let $(u,v)$ and $(w,z)$ be a pair of crossing edges of $G$, and let $C$ be the 4-cycle formed by the corresponding cycle edges. We claim that, in $G_p^+$, either the pair of vertices $\{u,v\}$, or the pair $\{w,z\}$, is such that both the vertices have outdegree at least one in $C$. Since two adjacent vertices in $C$ cannot both have outdegree two (in $C$), two non-adjacent vertices in $C$ must both have outdegree at least one.
\begin{figure}[t]
\centering
\begin{subfigure}[b]{.4\linewidth}
\centering
\includegraphics[scale=0.45,page=1]{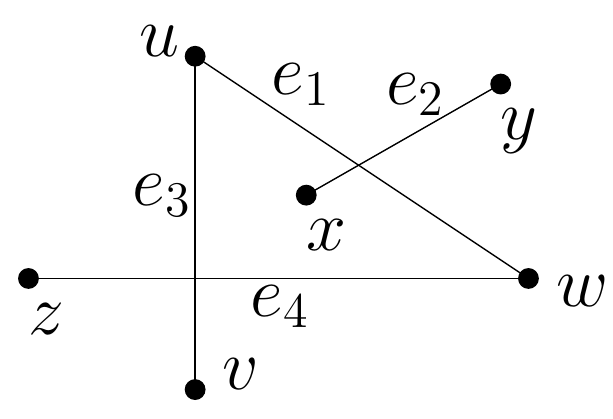}
\caption{}\label{fi:deg3-ub}
\end{subfigure}%
\begin{subfigure}[b]{.4\linewidth}
\centering
\includegraphics[scale=0.45,page=1]{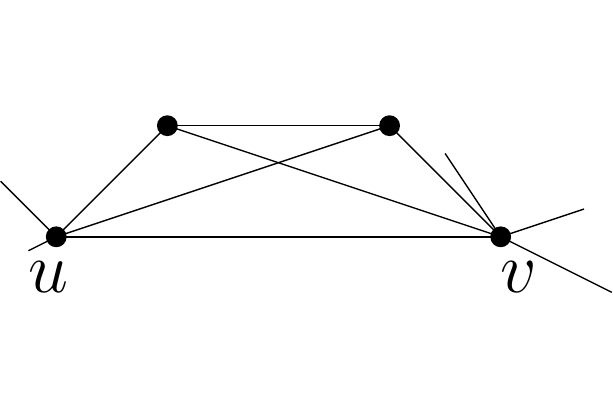}
\caption{}\label{fi:deg3-lb}
\end{subfigure}%
\caption{Illustration for the proof of Theorem~\ref{th:deg3}.}
\end{figure}
We now partition the edge set of $G$ as follows. For each pair of crossing edges $(u,v)$ and $(w,z)$ of $G$, we color red the edge connecting the pair, $\{u,v\}$ or $\{w,z\}$, for which the above claim holds (if it holds for both pairs, then we arbitrarily choose one). With this choice, each end-vertex of a red edge has one outgoing edge among the corresponding cycle edges. Recall that every vertex is incident to at most three outgoing edges in $G_p^+$. Also, by definition of NIC-plane graphs, no two pairs of crossing edges share a cycle edge. It follows that each vertex is incident to at most three red edges, as desired. The linear time complexity is a consequence of the fact that a $1$-plane graph has $O(n)$ edges (see e.g.~\cite{Suzuki2010}) and that Schnyder trees can be computed in $O(n)$ time~\cite{DBLP:conf/soda/Schnyder90}.

We conclude the proof by showing that there exist $3$-connected NIC-plane graphs such that the red graph of any edge partition has maximum vertex degree at least $3$. Let $G_p$ be a plane triangulation with $n \geq 7$ vertices. We construct the graph $G$ from $G_p$ as follows. We identify each edge $(u,v)$ of $G_p$ with a cycle edge of a kite (we can arbitrarily choose in which of the two faces of $(u,v)$ the kite will be placed), as shown in Figure~\ref{fi:deg3-lb}. The resulting graph is NIC-plane, since any two kites share at most one vertex, but it is not $3$-connected, since the end-vertices of each edge of $G_p$ are a separation pair of $G$. In order to make it $3$-connected, we triangulate all the faces of degree greater than three (without introducing parallel edges). Since $G_p$ has $3n-6$ edges and $G$ contains one pair of crossing edges for each edge of $G_p$, the red graph $G_R$ of any edge partition of $G$ has at least $3n-6$ edges. Also, each of these edges is incident to one vertex of $G_p$ (and to one vertex of $G \setminus G_p$). If the maximum vertex degree of $G_R$ is $k$, then it follows that $kn \geq 3n - 6$, which implies $k \geq  3 - \frac{6}{n}$, and, since $k$ is integer, $k \geq 3$ for $n \geq 7$.
\end{proof}

Given a $3$-connected NIC-plane graph $G$, we can compute a \kpartition{3} of $G$ by applying the algorithm of Theorem~\ref{th:deg3}, and then apply the algorithm of Lemma~$4$ in~\cite{ArXiv} to compute, in linear time, an embedding preserving OPVR $\gamma$ of $G$ with vertex complexity at most $6$. This proves the next corollary.

\begin{corollary}\label{co:opvr}
Let $G$ be a $3$-connected NIC-plane graph with $n$ vertices. There exists an $O(n)$-time algorithm that computes an embedding-preserving OPVR of $G$ with vertex complexity at most $6$.
\end{corollary}

\section{Degree-2 edge partitions}\label{se:deg2}

In this section we show that the problem of deciding whether a $1$-plane graph admits a \kpartition{2}, called DEGREE-2 EDGE PARTITION, is NP-complete. The hardness is proved by a reduction from PLANAR 3-SAT. We start by introducing some gadgets.

\begin{figure}[t]
\begin{subfigure}[b]{.5\linewidth}
\centering
\includegraphics[width=0.65\linewidth, page=2]{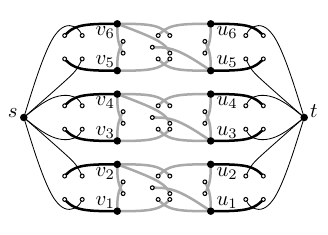}
\caption{}\label{fi:blob-1}
\end{subfigure}%
\begin{subfigure}[b]{.5\linewidth}
\centering
\includegraphics[width=0.65\linewidth, page=3]{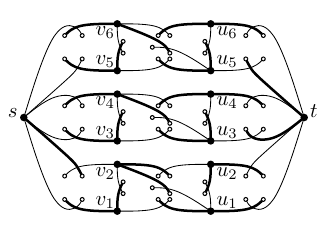}
\caption{}\label{fi:blob-2}
\end{subfigure}
\caption{(a) A blob gadget $B$. (b) A \kpartition{2} of $B$\label{fi:blob}}
\end{figure}

\paragraph{Blob gadget}  The \emph{blob gadget} is shown in Figure~\ref{fi:blob}. Its purpose is to force one specific vertex to have at least one red incident edge in any \kpartition{2}. It consists of two vertices $s$ and $t$ and two groups of six vertices each $u_1, \dots, u_6$ and $v_1, \dots v_6$ (these vertices are shown in black in Figure~\ref{fi:blob}). These vertices are then \emph{linked} by means of pairs of crossing edges. For all $i=1,2,3,4,5,6$ and $j=1,3,5$, $s$ is linked to $v_i$, $t$ is linked to $u_i$, $v_i$ is linked to $u_i$, $u_j$ is linked to $u_{j+1}$ and to $v_{j+1}$, and $v_j$ is linked to $v_{j+1}$.
The edge incident to $v_i$ or $u_i$ that crosses an edge incident to $s$ or $t$ is called an \emph{exterior edge}. All other edges incident to  $v_i$ or $u_i$ are called \emph{interior edges}. Exterior and interior edges are shown in bold and in gray, respectively, in Figure~\ref{fi:blob-1}.

\begin{lemma}\label{le:blob}
In any \kpartition{2} of a blob gadget there is at least one red edge incident to $s$.
\end{lemma}
\begin{proof}
We first observe that a blob gadget admits a \kpartition{2}. An example is shown in Figure~\ref{fi:blob-2}, where red edges are represented as bold edges.  
We use the notation of Figure~\ref{fi:blob-1}. Suppose, as a contradiction, that there exists a \kpartition{2} of a blob gadget such that there is no red edge incident to $s$. In this case the exterior edges of each $v_i$ (for $i=1,2,\dots,6$) must be red. Since there are at most two red edges incident to $t$, there exist two consecutive vertices $u_j $ and $u_{j+1}$ with $j$ odd whose exterior edges are red. Consider now the five pairs of crossing edges linking the vertices in the set $S=\{u_j, u_{j+1}, v_j, v_{j+1}\}$. These pairs contain ten interior edges each having one end-vertex in $S$.  Thus, there are five interior red edges incident to the vertices of $S$. Since $S$ has only four vertices, one of them, call it $w$, has two interior red edges. Since $w$ has also one exterior red edge it has three incident red edges, which is impossible in a \kpartition{2}.
\end{proof}

\begin{figure}[t]
\begin{subfigure}[b]{.5\linewidth}
\centering
\includegraphics[width=0.65\linewidth, page=2]{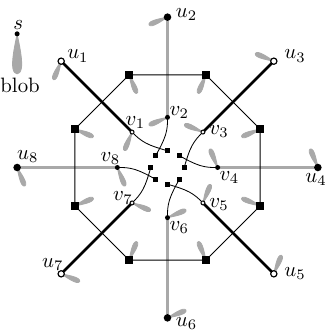}
\caption{}\label{fi:variable-1}
\end{subfigure}%
\begin{subfigure}[b]{.5\linewidth}
\centering
\includegraphics[width=0.65\linewidth, page=3]{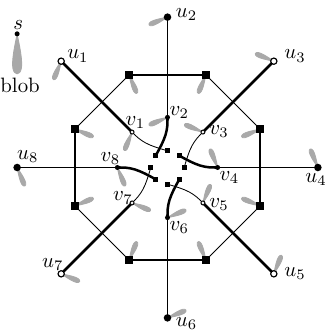}
\caption{}\label{fi:variable-2}
\end{subfigure}
\caption{(a) A variable gadget $V$ ($k=4$). (b) A \kpartition{2} of $V$.\label{fi:variable}}
\end{figure}

\paragraph{Variable gadget} The \emph{variable gadget} is shown in Figure~\ref{fi:variable-1}. Let $k$ be the number of clauses. The variable gadget contains a cycle $C$ of length $2k$ (squared vertices in Figure~\ref{fi:variable-1}) and each edge of $C$ is crossed by an edge $(u_i,v_i)$ (for $i=1,\dots,2k$). The vertices $v_i$ and $v_{i+1}$ with $i$ odd are linked by a pair of crossing edges. Also, a blob gadget is ``attached'' to the vertices of $C$, to the vertices $v_i$, and to the vertices $u_i$ (for $i=1,\dots,2k$) so that each such vertex coincides with the vertex $s$ of the blob. The edges $(u_i,v_i)$ with $i$ odd (bold in Figure~\ref{fi:variable-1}) are called \emph{variable edges}, while those with $i$ even (gray in Figure~\ref{fi:variable-1}) are called \emph{negated edges}. A vertex $u_i$ is a \emph{variable vertex} if it is an end-vertex of a variable edge, or a \emph{negated vertex} otherwise. A variable edge $(u_i,v_i)$ and a negated edge $(u_{i+1},v_{i+1})$ with $i$ odd are called \emph{twin edges}. 

\begin{lemma}\label{le:variable}
In any \kpartition{2} of a variable gadget, no variable edge and negated edge have the same color.
\end{lemma}
\begin{proof}
We first show that any two edges $(u_i,v_i)$ and $(u_{i+1},v_{i+1})$ ($i=1,2,\dots,2k-1$) cannot be both blue. If they are both blue, the two edges $e_i$ and $e_{i+1}$ crossing $(u_i,v_i)$ and $(u_{i+1},v_{i+1})$ are both red. Let $w$ be the vertex shared by $e_i$ and $e_{i+1}$. Since by Lemma~\ref{le:blob} vertex $w$ has an incident red edge in the blob attached to it, it has three incident red edges, which is impossible in a \kpartition{2}. 
We now prove that two twin edges $(u_i,v_i)$ and $(u_{i+1},v_{i+1})$ cannot be both red. If they are both red, vertices $v_i$ and $v_{i+1}$  are linked by a pair of crossing edges, one of which must be colored red. Suppose it is the one incident to $v_i$ and call it $e$. Then $v_i$ has three incident red edges: the edge $e$, the edge $(u_i,v_i)$, and, by Lemma~\ref{le:blob}, an edge of the blob attached to it. This is impossible in a \kpartition{2}. 
By the two arguments above, any pair of edges  $(u_i,v_i)$ and $(u_{i+1},v_{i+1})$ ($i=1,\dots,2k-1$) have distinct colors.    
\end{proof}

Based on Lemma~\ref{le:variable}, in any \kpartition{2} of a variable gadget, the variable edges all have the same color, while the negated edges all have the opposite color. Thus, a variable gadget has two possible \kpartition{2}{s}: the one shown in Figure~\ref{fi:variable-2} (red edges are  bold), and the one obtained by exchanging the colors. (In fact, the number of possible \kpartition{2}{s} is larger due to the fact that each blob gadget admits many different \kpartition{2}{s}. However the coloring of the blob gadget is not relevant for our discussion.)  We use these two configurations to encode the two truth values of the variable $x$ represented by the gadget. More precisely, when the variable edges are red, $x$ is \emph{true}; when they are blue, $x$ is \emph{false}. Since the negated edges have the other color with respect to the variable edges, we can use them to represent the negation $\neg{x}$ of $x$.

\begin{figure}[t]
\begin{subfigure}[b]{.5\linewidth}
\centering
\includegraphics[width=0.65\linewidth, page=1]{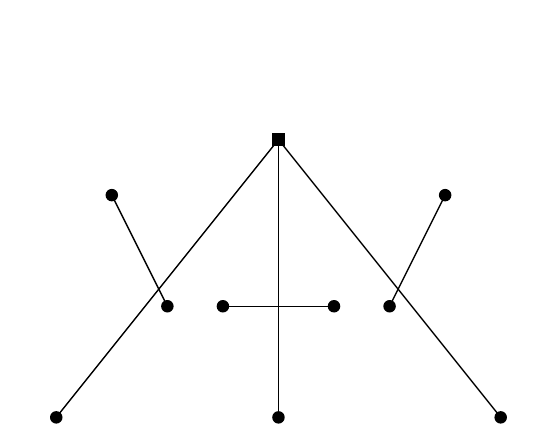}
\caption{}\label{fi:clause-1}
\end{subfigure}%
\begin{subfigure}[b]{.5\linewidth}
\centering
\includegraphics[width=0.65\linewidth, page=2]{clause.pdf}
\caption{}\label{fi:clause-2}
\end{subfigure}
\caption{(a) A clause gadget. (b) Connection to the variable gadgets.}\label{fi:clause}
\end{figure}

\paragraph{Clause gadget} The \emph{clause gadget} is shown in Figure~\ref{fi:clause-1}. It has three edges with a common end-vertex, called a \emph{clause vertex}, each crossed by a distinct edge. The three edges incident to the clause vertex are called \emph{false edges}, while the others are called \emph{true edges}. The idea is that each pair of crossing edges corresponds to a literal in the clause and each edge of the pair has an end-vertex in that literal's corresponding variable gadget. If the literal is true, the true edge is red otherwise the false edge is red. It is immediate to see that in any \kpartition{2} at least one true edge must be red. 

We now explain how to connect a clause gadget to the variable gadgets of its literals. Suppose that a clause $c$ contains a variable $x$ or its negation $\neg{x}$. Choose one false edge $e$ of the clause. Let $v$ be the vertex of $e$ different from the clause vertex, and let $w$ be an end-vertex of the edge $e'$ crossing $e$. Let $u_i$ ($1 \leq i \leq 2k$, $i$ odd) be a variable vertex of the variable gadget representing $x$ and let $u_{i+1}$ be the negated vertex following $u_{i}$. If $c$ contains the variable $x$, then $v$ is identified with $u_i$, while $w$ is identified with $u_{i+1}$. If $c$ contains the negation $\neg{x}$, then $v$ is identified with $u_{i+1}$, while $w$ is identified with $u_{i}$ (see Figure~\ref{fi:clause-2}). 
In this way the truth value of the variable is transmitted to the clause. Namely, suppose that $c$ contains $x$. If $x$ is true, the variable edge incident to $u_i$ is colored red; since, by Lemma~\ref{le:blob}, $u_i$ has another incident red edge in the blob gadget attached to it, the false edge $e$ cannot be red. On the other hand, the true edge $e'$ can be red because the negated edge incident to $u_{i+1}$ is blue. If $x$ is false, the true edge $e'$ cannot be red, while the false edge $e$ can be red. Clearly, if $c$ contains $\neg{x}$, we have the opposite situation. 

We now describe how to construct an instance $G_{\phi}$ of DEGREE-2 EDGE PARTITION from an instance $\phi$  of PLANAR 3-SAT. We create a variable gadget for each variable and a clause gadget for each clause. The clause gadgets are connected to the variable gadgets of their literals as explained (the number of pairs of variable/negated vertices is equal to the number of clauses and therefore each clause can use a different pair). We realize the connections between clauses and variable so to respect the planar embedding of $\phi$. An example is shown in Figure~\ref{fi:planar3sat}. Observe that all gadgets and $G_{\phi}$ are NIC-plane graphs.

\begin{figure}[t]
\centering
\begin{subfigure}[b]{.5\linewidth}
\centering
\includegraphics[width=0.8\linewidth, page=1]{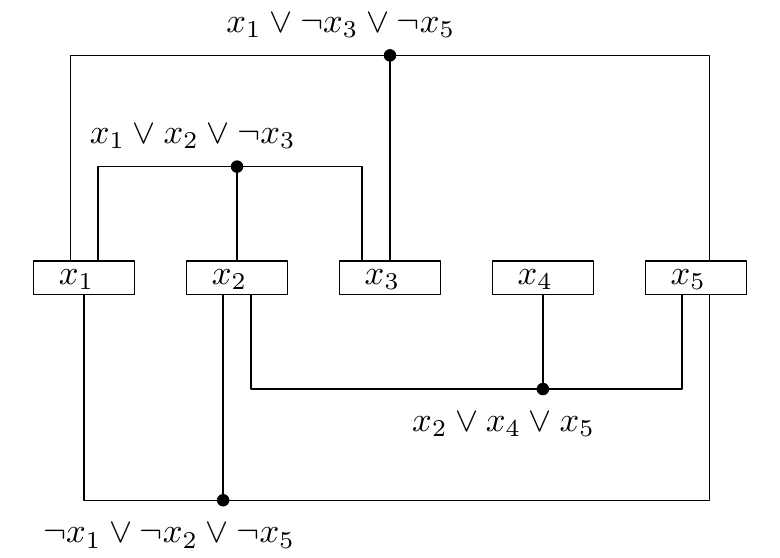}
\end{subfigure}%
\begin{subfigure}[b]{.5\linewidth}
\centering
\includegraphics[width=0.8\linewidth, page=2]{planar3sat.pdf}
\end{subfigure}
\caption{Left: An instance of PLANAR 3-SAT. Right: The corresponding instance of DEGREE-2 EDGE PARTITION.}\label{fi:planar3sat}
\end{figure}

\begin{theorem}\label{th:deg2}
It is NP-complete to decide whether a $1$-plane graph $G$ admits an edge partition such that the maximum vertex degree of the red graph is two, even if $G$ is NIC-plane.
\end{theorem}
\begin{proof}
Since one can verify a \kpartition{2} in polynomial time, DEGREE-2 EDGE PARTITION is in NP. Let $\phi$ be an instance of PLANAR 3-SAT and let $G_{\phi}$ be the corresponding instance of DEGREE-2 EDGE PARTITION as described above. We show that $\phi$ is satisfiable if and only if $G_{\phi}$ admits a \kpartition{2}.  
Assume first that $\phi$ is satisfiable. For each variable $x$, if it is true, we use for the variable gadget corresponding to $x$, the \kpartition{2} that assigns the red color to the variable edges; if it is false, we use the \kpartition{2} that assigns the red color to the negated edges. If a clause $c$ contains $x$ (respectively, $\neg{x}$), then we assign to the true edge of the gadget of $c$ connected to the variable gadget of $x$ the color red if and only if $x$ is true (respectively, false). As explained above, this coloring is feasible. Since each clause has at least one true literal, there are at most two red edges incident to the vertex shared by the false edges of that clause. 
Assume now that $G_{\phi}$ admits a \kpartition{2}. By the discussion above, it is easy to see that for each literal, the true edges of the clauses that contain that literal all have the same color, and this color is also assigned to the variable/negated edges corresponding to that literal. We assign to each literal the value true if the corresponding true edges are colored red, false otherwise. Moreover, since each vertex has at most two red incident edges, each clause must have at least one true literal, otherwise there would be three red edges incident to the clause vertex. This implies that $\phi$ is satisfiable.       
\end{proof}

\section{Degree-1 edge partitions}\label{se:deg1}
We show that whether a $1$-plane graph has a \kpartition{1} can be decided in quadratic time, by reducing the problem to 2-SAT. We show how to construct a $2$-CNF formula $\phi$, such that $G$ admits a \kpartition{1} if and only if $\phi$ can be satisfied. For each pair of crossing edges $p = \{e_1,e_2\}$, we assign to $e_1$ the variable $x_p$ and to $e_2$ the negation $\neg{x_p}$. For each vertex $v$ of $G$, let $y_{q_i}$, for $i=0,\dots,k_v-1$, be the $k_v \geq 0$ variables assigned to the edges incident to $v$ and that are crossed in $G$. For each pair $\{y_{q_i},y_{q_j}\}$, such that $0 \leq i \neq j < k_v$, we add to $\phi$ the clause $(y_{q_i} \Rightarrow \neg{y_{q_j}})$, which corresponds to $(\neg{y_{q_i}} \vee \neg{y_{q_j}})$ (observe that some literal $y_{q_i}$ of $\phi$ may be the negation of a variable). For example, consider the edges incident to vertex $v$ in the graph in Figure~\ref{fi:2sat}, we add to $\phi$ the clauses $(\neg{x} \vee y)$ $\wedge$  $(\neg{x} \vee \neg{z})$ $\wedge$  $(y \vee \neg{z})$. 
\begin{wrapfigure}{r}{29mm}
\centering
\includegraphics[scale=0.45]{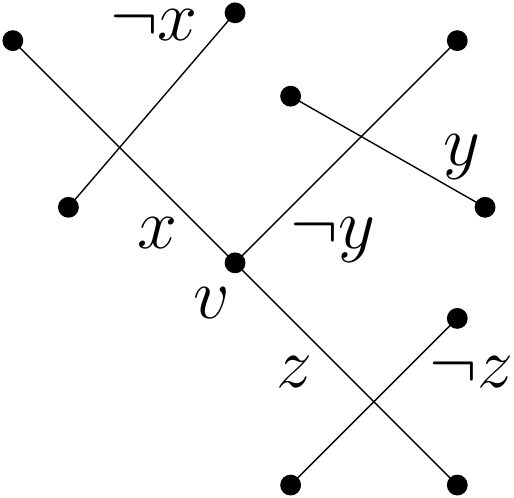}
\caption{\label{fi:2sat} Illustration for the proof of Theorem~\ref{th:deg1}.}
\end{wrapfigure}
Suppose that $G$ admits a \kpartition{1}, and consider the $2$-CNF formula $\phi$ described above. Let $p = \{e_1,e_2\}$ be a pair of crossing edges of $G$, such that $e_1$ is associated with the variable $x_p$, and $e_2$ with $\neg{x_p}$. If $e_1$ is red, then we set $x_p = true$, while if $e_2$ is red then we set $x_p = false$. In this way, we obtain a valid assignment for $\phi$. Indeed, let $v$ be a vertex of $G$, and let $y_{q_i}$, for $i=0,\dots,k_v-1$, be the $k_v \geq 0$ variables corresponding to the edges incident to $v$ and that are crossed in $G$. For each clause $(\neg{y_{q_i}} \vee \neg{y_{q_j}})$, we have that if $y_{q_i}$ is true (and thus $\neg{y_{q_i}}$ is false), then $y_{q_j}$ is false (and thus $\neg{y_{q_j}}$ is true), as otherwise $v$ would have degree two in the red graph. Suppose now that $\phi$ can be satisfied by some assignment $\mathcal A$. Let $p = \{e_1,e_2\}$ be a  pair of crossing edges of $G$, such that $e_1$ is associated with the variable $x_p$, and $e_2$ with $\neg{x_p}$. Color red $e_1$ if $x_p$ is true, while color red $e_2$ if $x_p$ is false. It is immediate to see that, since $\phi$ is satisfied by $\mathcal A$, every vertex is incident to at most one red edge. The above construction takes $O(n^2)$ time, since $1$-plane graphs have $O(n)$ edges, and thus $\phi$ contains $O(n^2)$ clauses. It follows that checking the satisfiability of $\phi$, and finding a valid assignment if it exists, can be done in $O(n^2)$ time~\cite{DBLP:journals/ipl/AspvallPT79}.  The next theorem summarizes the above discussion.

\begin{theorem}\label{th:deg1}
There exists an $O(n^2)$-time algorithm that tests whether an $n$-vertex $1$-plane graph $G$ admits a  \kpartition{1}, and, in the positive case, it computes one. 
\end{theorem}

\section{Conclusions and Open Problems}\label{se:conclusions}
We proved that every $3$-connected NIC-plane graph admits a \kpartition{3}, and hence an embedding preserving OPVR with vertex complexity at most six. Also, NIC-plane graphs may contain crossing configurations that imply a lower bound of one on the vertex complexity of any embedding preserving OPVR~\cite{DBLP:conf/compgeom/BiedlLM16}. It would be interesting to fill the gap between upper and lower bound. 

We proved that deciding whether a $1$-plane graph admits a \kpartition{2} is NP-complete, even when the graph is a NIC-plane graph. It is of interest to study the existence of approximation algorithms for the problem of computing \kpartition{k}{s} with minimum $k$.

We described a quadratic-time algorithm to decide whether a $1$-plane graph admits a \kpartition{1}. Is this problem solvable in linear time?

{\small \bibliography{paper}}
\bibliographystyle{abbrv}
\end{document}